\documentclass[11pt,reqno]{amsart}
\usepackage{graphicx,amsmath,amsthm,verbatim,amssymb,lineno,eucal,titletoc,enumerate} %eucal is "Euler" mathcal for the script N denoting an abelian network
\usepackage{bbm}
\usepackage{complexity}
\usepackage{hyperref}
\hypersetup{colorlinks}
\pdfoutput=1

\newcommand{\arxiv}[1]{{\tt \href{http://arxiv.org/abs/#1}{arXiv:#1}}}

\newcommand{\floor}[1]{\left\lfloor {#1} \right\rfloor}

\newcommand{\modulo}[1]{\quad (\mbox{mod }{#1})}
\newcommand{\modnospace}[1]{\; (\mbox{mod }{#1})}
\newcommand{\ghost}[1]{\textcolor{white}{#1}}

\newcommand{\old}[1]{}
\newcommand{\moniker}[1]{{\em (#1)}}
\newcommand{\dmoniker}[1]{(#1)}

\DeclareRobustCommand{\SkipTocEntry}[5]{}
\newcommand{\silentsec}[1]{\addtocontents{toc}{\SkipTocEntry}\section*{#1}}  % won't appear in TOC
\newcommand{\silentsubsec}[1]{\addtocontents{toc}{\SkipTocEntry}\subsection*{#1}}

\newtheorem{theorem}{Theorem}[section]

\newtheorem{lemma}[theorem]{Lemma}

\theoremstyle{remark}
\newtheorem{remark}[theorem]{Remark}

\numberwithin{counter}{section}

\theoremstyle{definition}
\newtheorem{definition}[theorem]{Definition}

\def\bb{\mathbf{b}}
\def\cc{\mathbf{c}}

\def\qq{\mathbf{q}}
\def\rr{\mathbf{r}}

\def\uu{\mathbf{u}}
\def\vv{\mathbf{v}}

\def\xx{\mathbf{x}}
\def\yy{\mathbf{y}}

\def\Proc{\CMcal{P}}  % regular mathcal
\def\Net{{\mathcal{N}}}  % "Euler" mathcal

\def\zero{\mathbf{0}}
\def\one{\mathbf{1}}

\def\basis{1}

\def\Sand{{\tt Sand}}

\def\N{\mathbb{N}}
\def\Z{\mathbb{Z}}

\def\R{\mathbb{R}}

\def\eps{\epsilon}

\begin{document}

\title[Abelian Networks]{Abelian Networks I. Foundations and Examples}

\author[Benjamin Bond and Lionel Levine]{Benjamin Bond and Lionel Levine}

\address{Benjamin Bond, Department of Mathematics, Stanford University, Stanford, California 94305. {\tt\url{http://stanford.edu/~benbond}}}
%\thanks{Supported by the MIT Undergraduate Research Opportunities Program.}

\address{Lionel Levine, Department of Mathematics, Cornell University, Ithaca, NY 14853. {\tt \url{http://www.math.cornell.edu/~levine}}}
%\thanks{Partly supported by an National Science Foundation postdoctoral fellowship and NSF DMS-1105960.}

%\date{December 30, 2014}
\date{February 3, 2016}


\begin{abstract}
In Deepak Dhar's model of abelian distributed processors, automata occupy the vertices of a graph and communicate via the edges. We show that two simple axioms ensure that the final output does not depend on the order in which the automata process their inputs.  A collection of automata obeying these axioms is called an \emph{abelian network}.
We prove a least action principle for abelian networks.  As an application, we show how abelian networks can solve certain linear and nonlinear integer programs asynchronously.  
In most previously studied abelian networks, the input alphabet of each automaton consists of a single letter; in contrast, we propose two non-unary examples of abelian networks: \emph{oil and water} and \emph{abelian mobile agents}.
\end{abstract} 

\maketitle
%\newpage
%\tableofcontents
%
\section{Introduction}

In recent years, it has become clear that certain interacting particle systems studied in combinatorics and statistical physics have a common underlying structure.  These systems are characterized by an \emph{abelian property} which says changing the order of certain interactions has no effect on the final state of the system.  
Up to this point, the tools used to study these systems -- least action principle, local-to-global principles, burning algorithm, transition monoids and critical groups -- have been developed piecemeal for each particular system.  Following Dhar \cite{Dha99a}, we aim to identify explicitly what these various systems have in common and exhibit them as special cases of what we call an \emph{abelian network}.

\old{
Intuition suggests that noncommutativity is a major source of dynamical richness and complexity.  Yet abelian networks produce surprisingly rich and intricate large-scale patterns from local rules \cite{Ost03,DSC09,DD13}. Just as in physics one infers from macroscopic observations the properties of microscopic particles, we would like to be able to infer from the large-scale behavior of a cellular automaton something about the local rules that generate that behavior.  In particular, are there certain large-scale features that can only be produced by \emph{noncommutative} local interactions?
By now a lot is known about the computational complexity of the abelian sandpile model \cite{GM97,MN99}; see \cite{MM11} for a recent compilation. 
The requirement that a distributed network produce the same output regardless of the order in which processors act would seem to place a severe restriction on the kinds of tasks it can perform.  Yet abelian networks can perform some highly nontrivial tasks, such as solving certain linear and nonlinear integer programs (see \textsection\ref{s.MIP}).  Are there other computational tasks that \emph{require} noncommutativity?

In this paper and its sequels, by defining abelian networks and exploring their fundamental properties, we hope to take a step toward making these questions precise and eventually answering them.  
}
After giving the formal definition of an abelian network in \textsection \ref{s.definition}, we survey a number of examples in \textsection \ref{s.examples}.  These include the well-studied sandpile and rotor networks as well as two \emph{non-unary} examples: oil and water, and abelian mobile agents. In \textsection \ref{s.leastaction} we prove a least action principle for abelian networks and explore some of its consequences.  One consequence is that ``local abelianness implies global abelianness'' (Lemma~\ref{l.localglobalabelian}).  Another is that abelian networks solve optimization problems of the following form: given a nondecreasing function $F : \N^k \to \N^k$, find the coordinatewise smallest vector $\uu \in \N^k$ such that $F(\uu) \leq \uu$ (if it exists).

This paper is the first in a series of three. In the sequel \cite{part2} we give conditions for a finite abelian network to halt on all inputs.  Such a network has a natural invariant attached to it, the \emph{critical group}, whose structure we investigate in \cite{part3}.  

\section{Definition of an abelian network}
\label{s.definition}

This section begins with the formal definition of an abelian network, which is based on Deepak Dhar's model of \emph{abelian distributed processors} \cite{Dha99a, Dha99b, Dha06}.  The term ``abelian network'' is convenient when one wants to refer to a collection of communicating processors as a single entity.  Some readers may wish to look at the examples in \textsection \ref{s.examples} before reading this section in detail.  

Let $G=(V,E)$ be a directed graph, which may have self-loops and multiple edges.  Associated to each vertex $v \in V$ is a \emph{processor} $\Proc_v$, which is an automaton with a single input port and multiple output ports, one for each edge $(v,u) \in E$.  Each processor reads the letters in its input port in first-in-first-out order.  

The processor $\Proc_v$ has an input alphabet $A_v$ and state space $Q_v$. These sets will usually be finite (but see \textsection\ref{s.oilwater} for an example with infinite state space).
We will always take the sets $A_v$ for $v \in V$ to be disjoint, so that a given letter belongs to the input alphabet just one processor.  No generality is lost by imposing this condition.

The behavior of the processor $\Proc_v$ is governed by a \emph{transition function} $T_v$ and \emph{message passing functions} $T_{(v,u)}$ associated to each edge $(v,u) \in E$.  Formally, these are maps
\begin{align*}
	T_v :  A_v \times Q_v \to Q_v  && \mbox{(new internal state)} \\
	T_{(v,u)} : A_v \times Q_v \to A_u^*
	&& \mbox{(letters sent from $v$ to $u$)}
\end{align*}
where $A_u^*$ denotes the free monoid of all finite words in the alphabet $A_u$.  We interpret these functions as follows.  If the processor $\Proc_v$ is in state $q$ and processes input $a$, then two things happen:
	\begin{enumerate}
	\label{i.prospective}
	\item[(1)] Processor $\Proc_v$ transitions to state $T_v(a,q)$; and
	\item[(2)] For each edge $(v,u) \in E$, processor $\Proc_u$ receives input $T_{(v,u)}(a,q)$.
	\end{enumerate}
% The choice of q rather than T_v(a,q) in step (2) is a convention.  It corresponds to the ``prospective'' convention in rotor-routing.

%Thus the entire network is specified by a transition function $T_v$ for each vertex $v \in V$, and a message passing function $T_{(v,u)}$ for each edge $(v,u) \in E$.  But 
If more than one $\Proc_v$ has inputs to process, then changing the order in which processors act may change the order of messages arriving at other processors.  Concerning this issue,  Dhar writes that 
\begin{quote}
``In many applications, especially in computer science, one considers such networks where the speed of the individual processors is unknown, and where the final state and outputs generated should not depend on these speeds. Then it is essential to construct protocols for processing such that the final result does not depend on the order at which messages arrive at a processor.'' \cite{Dha06}
\end{quote}
\hypertarget{i.wishlist}{Therefore we ask that the following aspects of the computation \emph{do not depend on the order in which individual processors act}}:
 
\begin{enumerate}
\item[(a)] The \textbf{halting status} (i.e., whether or not processing eventually stops).
\item[(b)] The \textbf{final states} of the processors.
\item[(c)] The \textbf{run time} (total number of letters processed by all $\Proc_v$).
\item[(d)] The \textbf{local run times} (number of letters processed by a given $\Proc_v$).
\item[(e)] The \textbf{detailed local run times} (number of times a given $\Proc_v$ processes a given letter $a \in A_v$).
\end{enumerate}

A priori it is not obvious that these goals are actually achievable by any nontrivial network.  In \textsection\ref{s.leastaction} we will see, however, that a simple local commutativity condition ensures all five goals are achieved.  
To state this condition, we extend the domain of $T_v$ and $T_{(v,u)}$ to $A_v^* \times Q_v$: if $w = a w'$ is a word in alphabet $A_v$ beginning with $a$, then set $T_v(w, q) = T_v(w', T_v(a, q))$
and
$T_{(v,u)}(w,q) = T_{(v,u)}(a,q) T_{(v,u)}(w', T_v(a,q))$, where the product denotes concatenation of words. For the empty word $\eps$, we set $T_v(\eps,q)=q$ and $T_{(v,u)}(\eps,q)=\eps$.

Let $\N^A$ be the free commutative monoid generated by $A$, and write $w \mapsto |w|$ for the natural map $A^* \to \N^A$. So $|w|$ is a vector with coordinates indexed by $A$, and its coordinate $|w|_a$ is the number of letters $a$ in the word $w$. In particular, words $w,w'$ satisfy $|w|=|w'|$ if and only if $w'$ is a permutation of $w$.

\begin{definition} \label{d.abelianprocessor}
 (Abelian Processor)
The processor $\Proc_v$ is called \emph{abelian} if for any words $w,w' \in A_v^*$ such that $|w| = |w'|$, we have for all $q \in Q_v$ and all edges $(v,u) \in E$
	\[ T_v(w,q) = T_v(w',q) \qquad \mbox{ and } \qquad 
	|T_{(v,u)}(w,q)| = |T_{(v,u)}(w',q)|.\]
That is, permuting the letters input to $\Proc_v$ does not change the resulting state of the processor $\Proc_v$, and may change each output word sent to $\Proc_u$ only by permuting its letters.
\end{definition}

A simple induction shows that if Definition~\ref{d.abelianprocessor} holds for words $w,w'$ of length $2$, then it holds in general; see \cite[Lemma 2.1]{HLW16}.

\begin{definition} \label{d.abeliannetwork}
 (Abelian Network)
An \emph{abelian network} on a directed graph $G=(V,E)$ is a collection of automata $\Net = (\Proc_v)_{v\in V}$ indexed by the vertices of $G$, such that each $\Proc_v$ is abelian.
\end{definition}

We make a few remarks about the definition: \medskip

1. The definition of an abelian network is \emph{local} in the sense that it involves checking a condition on each processor individually. As we will see, these local conditions imply a ``global'' abelian property (Lemma~\ref{l.localglobalabelian}). \smallskip

2. A processor $\Proc_v$ is called \emph{unary} if its alphabet $A_v$ has cardinality~$1$.  A unary processor is trivially abelian, and any network of unary processors is an abelian network.  Most of the examples of abelian networks studied so far are actually unary networks (an exception is the block-renormalized sandpile defined in \cite{Dha99a}).  Non-unary networks represent an interesting realm for future study.
The ``oil and water model'' defined in \textsection \ref{s.oilwater} is an example of an abelian network that is not a block-renormalized unary network.

%3. In interpreting the message passing functions, we have chosen the \emph{prospective} convention: the state of a vertex indicates what it will do next.  In the \emph{retrospective} convention, processor $\Proc_u$ receives input $T_{(v,u)}(a,T_v(a,q))$ instead, so that the state of a vertex indicates what it has just done.

\subsection{Comparison with cellular automata}
\label{s.cellular}

Cellular automata are traditionally studied on the grid $\Z^d$ or on other lattices, but they may be defined on any directed graph~$G$.  Indeed, we would like to suggest (see \textsection\ref{s.graphalg}) that the study of cellular automata on $G$ could be a fruitful means of revealing interesting graph-theoretic properties of~$G$. 

Abelian networks may be viewed as cellular automata enjoying the following two properties. \medskip

1. \textbf{Abelian networks can update asynchronously.} 
Traditional cellular automata update in parallel: at each time step, all cells \emph{simultaneously} update their states based on the states of their neighbors.   
Since perfect simultaneity is hard to achieve in practice, the physical significance of parallel updating cellular automata is open to debate.  Abelian networks do not require the kind of central control over timing needed to enforce simultaneous parallel updates, because they reach the same final state no matter in what order the updates occur.  \medskip

2. \textbf{Abelian networks do not rely on shared memory.}  
Implicit in the update rule of cellular automata is an unspecified mechanism by which each cell is kept informed of the states of its neighbors.  The lower-level interactions needed to facilitate this exchange of information in a physical implementation are absent from the model.  Abelian networks include these interactions by operating in a ``message passing'' framework instead of the ``shared memory'' framework of cellular automata: An individual processor in an abelian network cannot access the states of neighboring processors. It can only read the messages they send.

\section{Examples}\label{s.examples}

\begin{figure}[h]
\begin{center}
\includegraphics[height=.375\textheight]{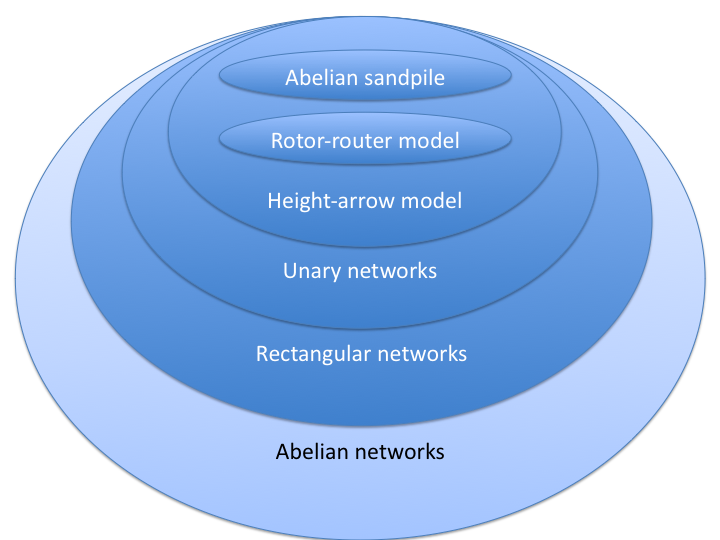}
\end{center}
\vspace{-4mm}
\caption{Venn diagram illustrating several classes of abelian networks.}
\label{fig:venn}
\end{figure}

%\silentsubsec{Sandpile networks} 
\subsection{Sandpile networks} 
\label{s.sandpile}
Figure~\ref{fig:venn} shows increasingly general classes of abelian networks.  
The oldest and most studied is the \emph{abelian sandpile model} \cite{BTW87, Dha90}, also called \emph{chip-firing} \cite{BLS91,Big99}.  Given a directed graph $G=(V,E)$, the processor at each vertex $v \in V$ has a one-letter input alphabet $A_v = \{v\}$ (we call the letter $v$ in order to keep the alphabets of different processors disjoint)
and state space $Q_v = \{0,1,\ldots,r_v-1\}$, where $r_v$ is the outdegree of $v$. The transition function is
	\[ T_v (q) = q+1 \modulo{r_v}. \]
(Formally we should write $T_v(v,q)$, but when $\# A_v =1$ we omit the redundant first argument.) The message passing functions are
	\[ T_{(v,u)} (q) = \begin{cases} \eps, & q<r_v-1 \\
							    u, & q=r_v-1
					\end{cases} \]
for each edge $(v,u) \in E$.  Here $\epsilon \in A^*$ denotes the empty word
%processing $\eps$ has no effect, (formally, $T_v(\eps,q)=q$ and $T_{(v,u)}(\eps,q)=\eps$), so 
(and passing the message $\eps$ is equivalent to passing nothing).
Thus each time the processor at vertex $v$ transitions from state $r_v-1$ to state $0$, it sends one letter to each of its out-neighbors (Figure~\ref{f.sandpile}). When this happens we say that vertex $v$ \emph{topples} (or ``fires'').

\begin{figure}[h]
\centering
\includegraphics[width=.25\textwidth]{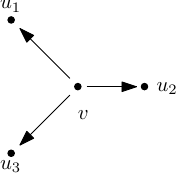} \\ \bigskip
\includegraphics[width=.6\textwidth]{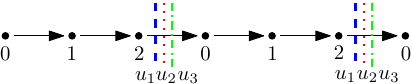} \\  \bigskip
\includegraphics[width=.6\textwidth]{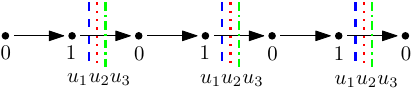}
\caption{Top: portion of graph $G$ showing a vertex $v$ and its outneighbors $u_1,u_2,u_3$.   
Middle: State diagram for $v$ in a \emph{sandpile network}. Dots represent states, arrows represent transitions when a letter is processed, and dashed vertical lines indicate when letters are sent to the neighbors.
Bottom: State diagram for the same $v$ in a \emph{toppling network} with $r_v=2$.}
\label{f.sandpile}
\end{figure}

Studies of pattern formation in sandpile networks include \cite{Ost03,DSC09,DD13}. The computational complexity of sandpile networks is investigated in \cite{GM97} (where a parallel update rule is required) and in \cite{MN99,MM11}, where the focus is on comparing the computational power of sandpile networks with underlying graph $\Z^d$ for different dimensions $d$.

%\bigskip 
%\noindent \textbf{Toppling networks} 
\subsection{Toppling networks}
\label{s.toppling}
These have the same transition and message passing functions as the sandpile networks above, but we allow the number of states $r_v$ (called the \emph{threshold} of vertex $v$) to be different from the outdegree of $v$.  These networks can be concretely realized in terms of ``chips'': If a vertex in state $q$ has $k$ letters in its input port, then we say that there are $q+k$ chips at that vertex.  
%Each vertex $v$ has a threshold $d_v$ 
%(in the case of sandpiles, $d_i$ is the out-degree of vertex $i$). 
When $v$ has at least $r_v$ chips, it can \emph{topple}, losing $r_v$ chips and sending one chip along each outgoing edge.   In a sandpile network the total number of chips is conserved, but in a toppling network, chips may be created (if $r_v$ is less than the outdegree of $v$, as in the last diagram of Figure~\ref{f.sandpile}) or destroyed (if $r_v$ is larger than the outdegree of $v$).  

Note that 
%when viewing a toppling network in terms of chips, 
some chips are ``latent'' in the sense that they are encoded by the internal states of the processors.  For example if a vertex $v$ with $r_v=2$ is in state~$0$, receives one chip and processes it, then the letter representing that chip is gone, but the internal state increases to~$1$ representing a latent chip at $v$.  If $v$ receives another chip and processes it, then its state returns to $0$ and it topples by sending one letter to each out-neighbor.

It is convenient to specify a toppling network by its \emph{Laplacian}, which is the $V\times V$ matrix $L$ with diagonal entries $L_{vv} = r_v-d_{vv}$ and off-diagonal entries $L_{uv} = -d_{uv}$.  Here $d_{uv}$ is the number of edges from $v$ to $u$ in the graph $G$. 
%The toppling network with no loops ($d_{vv}=0$) and Laplacian $L$ will be denoted $\Topp(L)$.

Sometimes it is useful to consider toppling networks where the number of chips at a vertex may become negative \cite{threshold}.  We can model this by enlarging the state space of each processor to include $-\N$; these additional states have transition function $T_v(q) = q+1$ and send no messages.  In \textsection \ref{s.MIP} we will see that these enlarged toppling networks solve certain integer programs.

%\silentsubsec{Sinks}
\subsection{Sinks and counters}
\label{s.sink}
It is common to consider the sandpile network $\Sand(G,s)$ with a \emph{sink} $s$, a vertex whose processor has only one state and never sends any messages.  If every vertex of $G$ has a directed path to the sink, then any finite input to $\Sand(G,s)$ will produce only finitely many topplings.  

The set of \emph{recurrent states} of a sandpile network with sink is in bijection with objects of interest in combinatorics such as oriented spanning trees and $G$-parking functions \cite{PS04}. Recurrent states of more general abelian networks are defined and studied in the sequel paper \cite{part3}.

%\silentsubsec{Counters}
%\subsection{Counters}
\label{s.counter}
A \emph{counter} is a unary processor with state space $\N$ and transition $T(q) = q+1$, which never sends any messages.  It behaves like a sink, but keeps track of how many letters it has received.
					
%\silentsubsec{Bootstrap percolation} 
\subsection{Bootstrap percolation}
In this simple model of crack formation, each vertex $v$ has a threshold $b_v$.
%, often chosen to be half its indegree (rounded up).  
Vertex $v$ becomes ``infected'' as soon as at least $b_v$ of its in-neighbors are infected.  Infected vertices remain infected forever.  A question that has received a lot of attention  \cite{Ent87,Hol03} due to its subtle scaling behavior is: What is the probability the entire graph becomes infected, if each vertex independently starts infected with probability $p$?  To realize bootstrap percolation as an abelian network, we take 
$A_v = \{v\}$ and $Q_v = \{0,1,\ldots,b_v\}$, with 
	$T_v(q) = \min(q+1, b_v)$
and
	\[ T_{(v,u)}(q) = \begin{cases} u, & q = b_v-1 \\
						       \epsilon, & q \neq b_v-1. \end{cases} \]
\old{ %%% non-unary version is more robust:
 $A_v$ to be the set $N_{in}(v)$ of in-neighbors of $v$ and $Q_v = 2^{N_{in}(v)}$ to be its power set.  The transition and message passing functions are given by
	\[ T_v(a,q) = q \cup \{a\} \]
and
	\[ T_{(v,u)}(a,q) = \begin{cases} \epsilon,  & \# q \neq b_v \\
							   v,  & \# q = b_v.
					\end{cases} \]
} %%%
State $b_v$ represents that $v$ is infected. Starting from a blank slate $q_v = 0$ for all $v$, the user sets up the initial condition by inputing $b_v$ letters $v$ to each initially infected vertex $v$.  The internal state $q$ of an initially healthy processor $\Proc_v$ keeps track of how many in-neighbors of $v$ are infected.  When this count reaches $b_v$, the processor $\Proc_v$ sends a letter to each out-neighbor of $v$ informing them that $v$ is now infected.
					
%\silentsubsec{Rotor networks} 
\subsection{Rotor networks}
\label{s.rotor}

A \emph{rotor} is a unary processor $\Proc_v$ that outputs exactly one letter for each letter input.  That is, 
for all $q \in Q_v$
	\begin{equation} \label{e.rotor} \sum_{(v,u) \in E} \sum_{a \in A_u} |T_{(v,u)}(q)|_a = 1.   \end{equation}
Inputting a single letter into a network of rotors yields an infinite walk $(v_n)_{n \geq 0}$, where vertex $v_n$ is the location of the single letter present after $n$ processings.  This walk has been termed \emph{stack walk} \cite{HP10} because of the following equivalent description (originating in \cite{DF91}). Each vertex $v$ has an infinite stack of cards, with each card labeled by a neighbor of $v$. The walker pulls the top card from the stack at her current location, steps to the indicated neighbor, throws away the card, and repeats.  The stack perspective features prominently in Wilson's algorithm for sampling uniformly from the set of spanning trees of a finite graph \cite{Wil96}.

In the special case that each stack is periodic, the stack walk has been studied under various names: In computer science it was introduced as a model of autonomous agents exploring a territory (``ant walk,'' \cite{WLB96}) and later studied as a means of broadcasting information through a network \cite{DFS08}.  In statistical physics it was proposed as a model of self-organized criticality (``Eulerian walkers,'' \cite{PDDK96}).  Propp called this case \emph{rotor walk} and proposed it as a way of derandomizing certain features of random walk \cite{Pro03,CS06,HP10,Pro10}.

\begin{figure}[h]
\centering
\includegraphics[width=.6\textwidth]{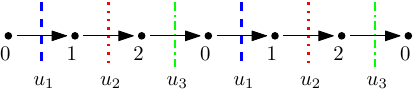}
\caption{State diagram for a vertex $v$ in a simple rotor network.  The out-neighbors $u_1,u_2,u_3$ of $v$ are served repeatedly in a fixed order.}
\label{f.rotor}
\end{figure}

Most commonly studied are the \emph{simple} rotor networks on a directed graph $G$, in which the out-neighbors of vertex $v$ are served repeatedly in a fixed order $u_1,\ldots,u_{d_v}$ (Figure~\ref{f.rotor}).  Formally, we set $Q_v = \{0,1,\ldots,d_v-1\}$,
% and $A_v = \{v\}$, 
with transition function $T_v(q) = q+1 \modnospace{d_v}$ and message passing functions
	\[ T_{(v,u_j)} (q) = \begin{cases} u_j, &q \equiv j-1 \pmod{d_v} \\
								\epsilon,  &q \not\equiv j-1 \pmod{d_v}.
					   \end{cases} \]

\silentsubsec{Rotor aggregation} 
Enlarge each state space $Q_v$ of a simple rotor network to include a transient state $-1$, which transitions to state $0$ but passes no message.  Starting with all processors in state $-1$, the effect is that each vertex ``absorbs'' the first letter it receives, and behaves like a rotor thereafter. If we input $n$ letters to one vertex $v_0$, then each letter performs a rotor walk starting from $v_0$ until reaching a site that has not yet been visited by any previous walk, where it gets absorbed.  Propp \cite{Pro03} proposed this model as a way of derandomizing a certain random growth process (internal DLA).
When the underlying graph is the square grid $\Z^2$, the resulting set of $n$ visited sites is very close to circular \cite{LP09}, and the final states of the processors display intricate patterns that are still not at all understood.

\begin{figure}[h]
\centering
\includegraphics[width=.6\textwidth]{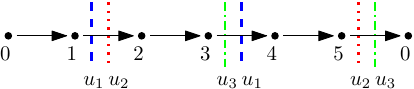} \\ \bigskip
\includegraphics[width=.6\textwidth]{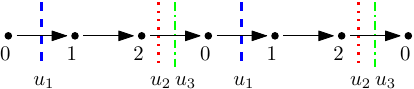}
\caption{Example state diagrams for a vertex $v$ in the height arrow model 
%with $d_v=3$ and $\tau_v=2$ 
(top) and Eriksson's periodically mutating game (bottom).}
\label{f.unary}
\end{figure}

%\silentsubsec{Unary networks}  
\subsection{Unary networks}
As shown in Figure~\ref{f.unary}, various other abelian processor state diagrams can be obtained by changing the locations of the vertical lines in Figure~\ref{f.rotor}. For example, Priezzhev, Dhar, Dhar and Krishnamurthy \cite{PDDK96} proposed a common generalization of rotor and sandpile networks, later studied by Dartois and Rossin~\cite{DR04} under the name \emph{height arrow model}. More generally, Diaconis and Fulton \cite{DF91} and Eriksson~\cite{Eri96} studied generalizations of chip-firing in which each vertex has a stack of instructions: When a vertex accumulates enough chips to follow the top instruction in its stack, it pops that instruction off the stack and follows it.  
These and all preceding examples are \emph{unary networks}, that is, abelian networks in which each alphabet $A_v$ has cardinality~$1$. Informally, a unary network on a graph $G$ is a system of local rules by which \emph{indistinguishable} chips move around on the vertices of $G$. 
%\medskip

Figures \ref{f.sandpile}--\ref{f.unary} are all one-dimensional because they diagram unary processors. In general, a processor with input alphabet $\{a_1,\ldots,a_d\}$ has a $d$-dimensional state diagram: states correspond to vectors in $\N^d$, and processing letter $a_i$ results in a transition from state $q$ to state $q+e_i$ where $e_1,\ldots,e_d$ are the standard basis vectors. The vertical bars that indicate message passing in Figures \ref{f.sandpile}--\ref{f.unary} become $(d-1)$-dimensional plaquettes, each labeled by a letter to be passed. 
%The line segment from $q$ to $q+e_i$ may cross zero, one or more surfaces whose labels belong to $A_u$; these labels in order form the word passed along $(v,u)$ when $\Proc_v$ in state $q$ processes letter $a_i$. 
A visual manifestation of the abelian property is that these plaquettes join up into surfaces of ``negative slope'': For example, beginning at the left side of Figure~\ref{f.agents} each red or blue message line ($d-1=1$) consists of only downward and rightward steps, ensuring that for any two states $q,q'$ any two paths of upward and rightward steps from $q$ to $q'$ cross the same set of message lines.

The next two sections discuss non-unary examples.

%\silentsubsec{Abelian mobile agents}
\subsection{Abelian mobile agents}
\label{s.mobile}

In the spirit of \cite{WLB96}, one could replace the messages in our definition of abelian networks by mobile agents each of which is an automaton. As a function of its own internal state $a$ and the state $q$ of the vertex~$v$ it currently occupies, an agent acts by doing three things:
	\begin{enumerate} 
	\item it changes its own state to $S_v(a,q)$; and
	\item it changes the state of~$v$ to $T_v(a,q)$; and
	\item it moves to a neighboring vertex $U_v(a,q)$.
	\end{enumerate}
% note: simpler with this convention (prospective)
Two or more agents may occupy the same vertex, in which case we require that the outcome of their actions is the same regardless of the order in which they act.  For purposes of deciding whether two outcomes are the same, we regard agents with the same internal state and location as indistinguishable.  

This model may appear to lie outside our framework of abelian networks, because the computation is located in the moving agents (who carry their internal states with them) instead of in the static processors.  However, it has identical behavior to the following abelian network. Denoting by $M$ the set of possible agent internal states (we could call them ``moods'' to distinguish them from the internal states of the vertices), let each vertex have input alphabet $M$ (technically, we should take the input alphabet of $v$ to be $\{v\} \times M$ to abide by our convention that the input alphabets are disjoint) with transition function $M \times Q_v \to Q_v$ sending $(a,q) \mapsto T_v(a,q)$, and message passing function $M \times Q_v \to M \cup \{\eps\}$ given by
	\[ T_{(v,u)}(a,q) = \begin{cases} S_v(a,q) & \text{if } u = U_v(a,q) \\
							\eps & \text{else.}  \end{cases} \]
% also prospective

Abelian mobile agents generalize the rotor networks (\textsection \ref{s.rotor}) by dropping the requirement that processors be unary.

\begin{figure}[h]
\begin{center}
\includegraphics[width=0.6\textwidth]{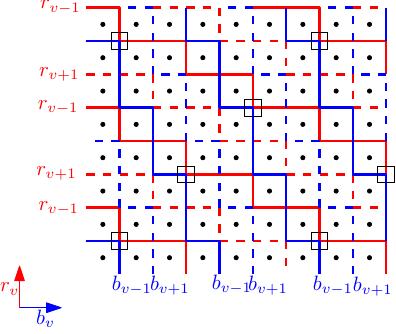}
\end{center}
\caption{Abelian mobile agents: Example state diagram for a processor $\Proc_v$ in a network whose underlying graph is $\Z$. The two dimensions correspond to the two letters in the input alphabet $A_v = \{r_v,b_v\}$, representing a red or blue agent at vertex $v$. Each black dot represents a state $q \in Q_v$. When processor $\Proc_v$ in state $q$ processes a letter, it transitions to state $q + (0,1)$ or $q+(1,0)$ depending on whether the letter was $r_v$ or $b_v$. The solid and dashed colored lines indicate message passing: Each line is labeled by one of the letters $r_{v-1}, r_{v+1}, b_{v-1}, b_{v+1}$, representing that the agent may step either left or right from $v$ and may change color.
 The small black boxes highlight the lattice of periodicity, generated by $(6,0)$ and $(2,2)$. The size of the state space $\# Q_v$ is the index of the lattice, which is $12$ in this example.
 }
\label{f.agents}
\end{figure} 

The defining property of abelian mobile agents is that each processor sends exactly one letter for each letter received.  In Figure~\ref{f.agents} this property is apparent from the fact that each segment of the square grid lies on exactly one message line.  The caption is written from the processor's point of view. From the agent's point of view, it could read as follows. When an agent arrives at vertex $v$, she updates the internal state $q$ of $\Proc_v$ depending on her mood: if her mood is red then she increments $q$ by $(0,1)$ and if blue then she increments $q$ by $(1,0)$. The old and new states are adjacent black dots in the figure, separated by exactly one message line. The agent updates her mood to red or blue according to the color of this line, and she moves to vertex $v-1$ or $v+1$ according to whether this line is solid or dashed. 

For example, supposing the initial state of $\Proc_v$ is $(0,0)$ (the bottom left dot) and there is one red and one blue agent at $v$. This means $\Proc_v$ has two letters in its input port, $r_v$ and $b_v$.  If the blue agent acts first, then $\Proc_v$ transitions to state $(1,0)$ and outputs $b_{v-1}$, representing that the blue agent steps to $v-1$ and remains blue. If now the red agent at $v$ acts, then $\Proc_v$ transitions from state $(1,0)$ to $(1,1)$ and outputs $r_{v-1}$, representing that the red agent steps to $v-1$ and remains red.  Note that if the agents had acted in the opposite order, then both would have changed color, so the net result is the same: one red and one blue agent at $v-1$.
	
%\silentsubsec{Oil and water model}
\subsection{Oil and water model}
\label{s.oilwater}

This is a non-unary generalization of sandpiles, inspired by Paul Tseng's asynchronous algorithm for solving certain linear programs \cite{Tse90}. 
Each edge of $G$ is marked either as an oil edge or a water edge.
%(if $G$ is a multigraph, then parallel edges need not receive the same marking) 
When a vertex topples, it sends out one oil chip along each outgoing oil edge and also one water chip along each outgoing water edge. The interaction between oil and water is that a vertex is permitted to topple if and only if sufficiently many chips of \emph{both} types are present at that vertex: that is, the number of oil chips present must be at least the number of outgoing oil edges, and the number of water chips present must be at least the number of outgoing water edges.  
%Thus the oil chips obey sandpile dynamics on $G_{oil}$, but with topplings allowed only if sufficiently many water chips are present; and vice versa.

\begin{figure}[h]
\centering
\includegraphics[width=.27\textwidth]{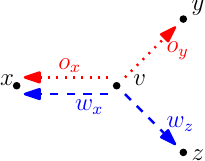} \\ \bigskip
\includegraphics[width=.65\textwidth]{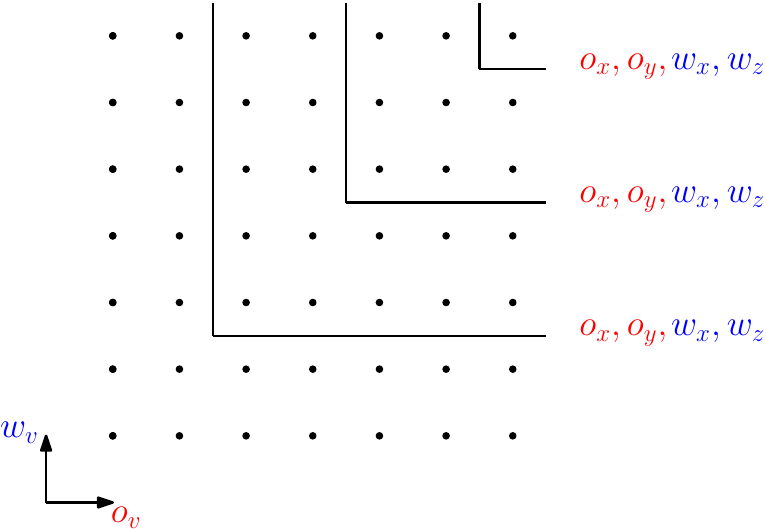}
\caption{Example state diagram for the oil and water model. Top: Vertex $v$ has outgoing oil edges to $x$ and $y$, and water edges to $x$ and $z$. Bottom: each dot represents a state in $Q_v = \N \times \N$, with the origin at lower left.  A toppling occurs each time the state transition crosses one of the bent lines (for example, by processing an oil $o_v$ in state $(1,2)$, resulting in transition to state $(2,2)$). Since $v$ has outdegree $2$ in both the oil graph and the water graph, the bent lines run to the left of columns whose $x$-coordinate is divisible by $2$, and below rows whose $y$-coordinate is divisible by $2$.}
\label{f.oilwater}
\end{figure}

Unlike most of the preceding examples, oil and water can not be realized with a finite state space $Q_v$, because an arbitrary number of oil chips could accumulate at $v$ and be unable to topple if no water chips are present.  We set  $Q_v = \N \times \N$ and $A_v = \{o_v, w_v\}$ representing an oil or water chip at vertex $v$, with transition function
  	\[ T_v(o_v,q) = q+(0,1), \qquad T_v(w_v,q) = q+(1,0). \]
The internal state of the processor at $v$ is a vector $q =(q_{oil}, q_{water})$ keeping track of the total number chips of each type it has received (Figure~\ref{f.oilwater}).  Stochastic versions of the oil and water model are studied in \cite{APR09, oil}.

\subsection{Stochastic abelian networks}

In a stochastic abelian network, we allow the transition functions 
to depend on a probability space $\Omega$:
\begin{align*}
	T_v :  A_v \times Q_v \times \Omega \to Q_v  && \mbox{(new internal state)} \\
	T_{(v,u)} : A_v \times Q_v \ghost{\;\times\, \Omega} \to A_u^*
	&& \mbox{(letters sent from $v$ to $u$)}
\end{align*}
A variety of models in statistical mechanics --- including classical Markov chains and branching processes, branching random walk, certain directed edge-reinforced walks, internal DLA \cite{DF91}, the Oslo model \cite{Fre93}, the abelian Manna model \cite{Dha99c}, excited walk \cite{BW03}, the Kesten-Sidoravicius infection model \cite{KS05, KS08}, two-component sandpiles and related models derived from abelian algebras \cite{AR08, APR09}, activated random walkers \cite{DRS10}, stochastic sandpiles \cite{RS12,CMS13}, and low-discrepancy random stack \cite{FL13} --- can all be realized as stochastic abelian networks.  In at least one case \cite{RS12} the abelian nature of the model enabled a major breakthrough in proving the existence of a phase transition.  Stochastic abelian networks are beyond the scope of the present paper and will be treated in a sequel.

\section{Least action principle}
\label{s.leastaction}

Our first aim is to prove a least action principle for abelian networks, Lemma~\ref{l.leastaction}. This principle says --- in a sense to be made precise --- that each processor in an abelian network performs the minimum amount of work possible to remove all letters from the network.  Various special cases of the least action principle to particular abelian networks have enabled a flurry of recent progress: bounds on the growth rate of sandpiles \cite{FLP10}, exact shape theorems for rotor aggregation \cite{KL10,HS11}, proof of a phase transition for activated random walkers \cite{RS12}, and a fast simulation algorithm for growth models \cite{FL13}.  The least action principle was also the starting point for the recent breakthrough by Pegden and Smart \cite{PS13} showing existence of the abelian sandpile scaling limit.

The proof of the least action principle follows Diaconis and Fulton \cite[Theorem~4.1]{DF91}.  Our observation  is that their proof actually shows something more general: it applies to any abelian network.  Moreover, as noted in \cite{Gab94,FLP10,RS12}, the proof applies even to executions that are complete but not legal.  To explain the last point requires a few definitions.

Let $\Net$ be an abelian network with underlying graph $G=(V,E)$, total state space~$Q = \prod Q_v$ and total alphabet~$A = \sqcup A_v$. In this section we do not place any finiteness restrictions on $\Net$: the underlying graph may be finite or infinite, and the state space $Q_v$ and alphabet $A_v$ of each processor may be finite or infinite.  

We may view the entire network $\Net$ as a single automaton with alphabet $A$ and state space $\Z^A \times Q$.  For its states we will use the notation $\xx.\qq$, where $\xx \in \Z^A$ and $\qq \in Q$.  If $\xx \in \N^A$ the state $\xx.\qq$ corresponds to the configuration of the network $\Net$ such that
	\begin{itemize} 
	\item For each $a \in A$, there are $\xx_a$ letters $a$ waiting to be processed; and
	\item For each $v \in V$, the processor at vertex $v$ is in state $q_v$.
	\end{itemize}
Formally, $\xx.\qq$ is just an alternative notation for the ordered pair $(\xx,\qq)$.  
The decimal point in $\xx.\qq$ is intended to evoke the intuition that the internal states $\qq$ of the processors represent latent ``fractional'' messages. Note that $\xx$ indicates only the \emph{number} of letters present of each type. We may think of $\xx$ as a collection of piles of letters, one pile for each vertex: Recalling that the alphabets $A_v$ are disjoint, the $\xx_a$ letters $a$ are in the pile of the unique vertex $v$ such that $a \in A_v$.  
\old{
Allowing $\xx$ to have negative coordinates is a useful device that enables the least action principle (Lemma~\ref{l.leastaction} below).
}
% Indeed, one of our goals is to show that the order does not matter (Theorem~\ref{t.wishlist}). 

In what follows it may be helpful to imagine that some entity, the \emph{executor},
%(who may be the user or an adversary), 
chooses the order in which letters are processed. Formally, these choices are encoded by a word $w = w_1 \cdots w_r$ where each letter $w_i \in A$, instructing the network first to process letter $w_1$, then $w_2$, etc.
We are going to allow for the possibility that the executor makes an ``illegal'' move by
choosing to process some letter (say $a$) even if it is not in the pile, resulting in the coordinate $\xx_a$ becoming negative. 
%The least action principle will say that an executor whose goal is to minimize the number of letters processed can not gain from such a choice.

For $v \in V$ and $a \in A_v$, denote by $t_a : Q \to Q$ the map
	\[ t_a(\qq)_u = \begin{cases} T_v(a,q_v), & u=v \\
						    q_u, & u \neq v \end{cases} \]
where $T_v$ is the transition function of vertex $v$ (defined in \textsection\ref{s.definition}).
The effect of processing one letter $a$ on the pair $\xx.\qq$ is described by a  map $\pi_a : \Z^A \times Q \to \Z^A \times Q$, namely
	\begin{equation} \label{e.onepass} \pi_a(\xx.\qq) = (\xx - \basis_a + \mathbf{N}(a,q_v)).t_a(\qq) \end{equation}
where $(\basis_{a})_b$ is $1$ if $a=b$ and $0$ otherwise; and 
$\mathbf{N}(a,q_v)_b$ is the number of $b$'s produced when processor $\Proc_v$ in state $q_v$ processes the letter $a$.  
In other words,
	\[ \mathbf{N}(a,q_v) = \sum_{e} |T_e(a,q_v)| \] 
where $T_e$ is the message passing function of edge $e$, and the sum is over all outgoing edges $e$ from $v$ (both sides are vectors in $\Z^A$).

Having defined $\pi_a$ for letters $a$, we define $\pi_w$ for a word $w = w_1\cdots w_r \in A^*$ as the composition $\pi_{w_r} \circ \cdots \circ \pi_{w_1}$. 
To generalize equation \eqref{e.onepass}, we extend the domain of $\mathbf{N}$ to $A^* \times Q$ as follows.  Let $\qq^{i-1} =  (t_{w_{i-1}} \circ \cdots \circ t_{w_{1}})\qq$ and let
	\[ \mathbf{N}(w,\qq) := \sum_{i=1}^r \mathbf{N}(w_i, \qq^{i-1}_{v(i)}) \]
where $v(i)$ is the unique vertex such that $w_i \in A_{v(i)}$.  Note that if $a\in A_v$ and $b \in A_u$ for $v\neq u$, then 
	\begin{equation} \label{e.locality} \mathbf{N}(ab,\qq) = \mathbf{N}(a,\qq)+\mathbf{N}(b,\qq) \end{equation}
since $t_a$ acts by identity on $Q_u$ and $t_b$ acts by identity on $Q_v$.

Recall that $|w| \in \N^A$ and $|w|_a$ is the number of occurrences of letter $a$ in the word $w$.  
From the definition of $\pi_a$ we have by induction on $r$
	\begin{equation} \label{e.piformula} \pi_w(\xx.\qq) = (\xx - |w| + \mathbf{N}(w,\qq)).t_w(\qq) \end{equation}
where $t_w := t_{w_r} \circ \cdots \circ t_{w_1}$.

In the next lemma and throughout this paper, inequalities on vectors are coordinatewise.

\begin{lemma} 
\label{l.monotonicity}
\moniker{Monotonicity}
For $w,w' \in A^*$ and $\qq \in Q$, if $|w| \leq |w'|$, then $\mathbf{N}(w,\qq) \leq \mathbf{N}(w',\qq)$.
\end{lemma}

\begin{proof}
For a vertex $v \in V$ let $p_v : A^* \to A_v^*$ be the monoid homomorphism defined by $p_v(a) = a$ for $a \in A_v$ and $p_v(a) = \epsilon$ (the empty word) for $a \notin A_v$.
Equation \eqref{e.locality} implies that
	\[ \mathbf{N}(w,\qq) = \sum_{v \in V} \mathbf{N}(p_v(w),\qq), \]
so it suffices to prove the lemma for $w,w' \in A_v^*$ for each $v \in V$.

Fix $v \in V$ and $w,w' \in A_v^*$ with $|w| \leq |w'|$. Then there is a word $w''$ such that $|ww''| = |w'|$.  Given a letter $a \in A_u$, if $(v,u) \notin E$ then $\mathbf{N}(w,\qq)_a = \mathbf{N}(w',\qq)_a = 0$.  If $(v,u) \in E$, then since $\Proc_v$ is an abelian processor,
	\begin{align*} \mathbf{N}(w',\qq)_a = |T_{(v,u)}(w',q_v)|_a &= |T_{(v,u)}(ww'',q_v)|_a 
		\\ &= |T_{(v,u)}(w,q_v)|_a + |T_{(v,u)}(w'',T_v(w,q_v))|_a.
%		\\ &= \mathbf{N}(w,\qq)_a + \mathbf{N}(w'',\qq'').
	\end{align*}
The first term on the right side equals $\mathbf{N}(w,\qq)_a$, and the remaining term is nonnegative, completing the proof.
\end{proof}

\begin{lemma}
\label{l.piscommute}
For $w,w' \in A^*$, if $|w|_a = |w'|_a$ for all $a \in A$, then $\pi_w = \pi_{w'}$.
\end{lemma}

\begin{proof}
Suppose $|w|=|w'|$.  Then for any $\qq \in Q$ we have $\mathbf{N}(w,\qq)=\mathbf{N}(w',\qq)$ by Lemma~\ref{l.monotonicity}.
Since $t_a$ and $t_b$ commute for all $a,b \in A$, we have $t_w(\qq)=t_{w'}(\qq)$.  Hence the right side of \eqref{e.piformula} is unchanged by substituting $w'$ for $w$.
\end{proof}

\subsection{Legal and complete executions}

An \emph{execution} is a word $w \in A^*$.  It prescribes an order in which letters in the network are to be processed.  For simplicity, we consider only finite executions in the present paper, but we remark that infinite executions (and non-sequential execution procedures) are also of interest \cite{FMR09}.
%in particular, if $V$ is infinite then there may be executions $w= w_1 w_2 \cdots$ which remove all letters from the network and are locally finite in the sense that each letter $a$ occurs only finitely many times in $w$. 

Fix an initial state $\xx.\qq$. The letter $a \in A$ is called a \emph{legal move} from $\xx.\qq$ if $\xx_a \geq 1$.  An execution $w = w_1 \cdots w_r$ is called \emph{legal} for $\xx.\qq$ if $w_{i}$ is a legal move from $\pi_{w_1\cdots w_{i-1}}(\xx.\qq)$ for all $i=1,\ldots,r$. An execution $w$ is called \emph{complete} for $\xx.\qq$ if $\pi_w(\xx.\qq)=\yy.\qq'$ for some $\qq' \in Q$ and $\yy \in \Z^A$ with $\yy_a \leq 0$ for all $a \in A$. We emphasize that a complete execution need not be legal.

\begin{lemma} \label{l.leastaction}
\moniker{Least Action Principle}
If $w$ is legal for $\xx.\qq$ and $w'$ is complete for $\xx.\qq$, then $|w|_a \leq |w'|_a$ for all $a \in A$.
%If $w = w_1 \cdots w_r$ is legal for $\xx.\qq$ and $w' = w'_1 \cdots w'_s$ is complete for $\xx.\qq$, then $|w| \leq |w'|$ and $r \leq s$.
\end{lemma}

\begin{proof}
%Noting that $r = \sum_{a \in A} |w|_a$ and $s = \sum_{a \in A} |w'|_a$, it suffices to prove $|w| \leq |w'|$.
Let $w=w_1\cdots w_r$.
Supposing for a contradiction that $|w| \not\leq |w'|$, let $i$ be the smallest index such that $|w_1 \cdots w_i| \not \leq |w'|$.  Let $u = w_1 \cdots w_{i-1}$ and $a = w_i$.  By the choice of $i$ we have $|u|_a = |w'|_a$, and $|u|_b \leq |w'|_b$ for all $b \neq a$.  Since $w$ is legal for $\xx.\qq$, at least one letter $a$ is present in $\pi_\uu (\xx.\qq)$, so by \eqref{e.piformula} and Lemma~\ref{l.monotonicity}
	\begin{align*} 1 		&\leq \xx_a - |u|_a + \mathbf{N}(u,\qq)_a \\
					  &\leq \xx_a - |w'|_a + \mathbf{N}(w',\qq)_a.
	%\\				  &= \xx'^s_a.
	 \end{align*}
Since $w'$ is complete for $\xx.\qq$, the right side is $\leq 0$ by \eqref{e.piformula}, which yields the required contradiction.
\end{proof}

\subsection{Halting dichotomy}

\begin{lemma} \label{l.dichotomy}
\moniker{Halting Dichotomy}
For a given initial state $\qq$ and input $\xx$ to an abelian network $\Net$, either
\begin{enumerate}
\item There does not exist a finite complete execution for $\xx.\qq$; or
\item Every legal execution for $\xx.\qq$ is finite, there exists a complete legal execution for $\xx.\qq$, and any two complete legal executions $w,w'$ for $\xx.\qq$ satisfy $|w|=|w'|$. 
\end{enumerate}
\end{lemma}

\begin{proof}
If there exists a finite complete execution, say of length $s$, then every legal execution has length $\leq s$ by Lemma~\ref{l.leastaction}.  The empty word is a legal execution, and any legal execution of maximal length is complete (else it could be extended by a legal move).
 If $w$ and $w'$ are complete legal executions, then $|w| \leq |w'| \leq |w|$ by Lemma~\ref{l.leastaction}.
\end{proof}

Note that in case (1) any finite legal execution $w$ can be extended by a legal move: since $w$ is not complete, there is some letter $a\in A$ such that $wa$ is legal.  So in this case there is an infinite word $a_1 a_2 \cdots$ such that $a_1 \cdots a_n$ is a legal execution for all $n \geq 1$. The \emph{halting problem} for abelian networks asks, given $\Net$, $\xx$ and $\qq$, whether (1) or (2) of Lemma~\ref{l.dichotomy} is the case.  In case (2) we say that $\Net$ \emph{halts} on input $\xx.\qq$.  In the sequel \cite{part2} we characterize the finite abelian networks that halt on all inputs.

\subsection{Global abelianness}

\begin{definition}
\label{d.odometer}
\dmoniker{Odometer}
If $\Net$ halts on input $\xx.\qq$, we denote by $[\xx.\qq]_a = |w|_a$ the total number of letters $a$ processed during a complete legal execution $w$ of $\xx.\qq$.  The vector $[\xx.\qq] \in \N^A$ is called the \emph{odometer} of $\xx.\qq$.  By Lemma~\ref{l.dichotomy}, the odometer does not depend on the choice of complete legal execution $w$.
\end{definition}

No messages remain at the end of a complete legal execution $w$, so the network ends in state $\pi_w(\xx.\qq) = \zero.t_w(\qq)$.  Hence by \eqref{e.piformula}, the odometer can be written as
	\[ [\xx.\qq] = |w| = \xx + \mathbf{N}(w,\qq) \]
which simply says that the total number of letters processed (of each type $a\in A$) is the sum of the letters input and the letters produced by message passing.
The coordinates of the odometer are the ``detailed local run times'' from \textsection\ref{s.definition}. We can summarize our progress so far in the following theorem.

\begin{theorem}
\label{t.wishlist}
Abelian networks have properties \hyperlink{i.wishlist}{(a)--(e)} from \textsection\ref{s.definition}.
\end{theorem}

\begin{proof}
By Lemma~\ref{l.dichotomy} the halting status does not depend on the execution, which verifies item (a). Moreover for a given $\Net,\xx,\qq$ any two complete legal executions  have the same odometer, which verifies items (c)--(e). The odometer and initial state $\qq$ determine the final state $t_w(\qq)$, which verifies (b).
\end{proof}

The next lemma illustrates a general theme of \emph{local-to-global principles} in abelian networks.  Suppose we are given a partition $V = I \sqcup M \sqcup O$ of the vertex set into ``input'', ``mediating'' and ``output'' nodes,
and that the output nodes never send messages (for example, the processor at each output node could be a counter, \textsection \ref{s.counter}). We allow the possibility that $M$ and/or $O$ is empty.
If $\Net$ halts on all inputs, then we can regard the induced subnetwork $(\Proc_v)_{v \in I \cup M}$ of non-output nodes as a single processor $\Proc_{I,M}$ with input alphabet $A_I := \sqcup_{v \in I} A_v$, state space $Q_{I \cup M} := \prod_{v \in I \cup M} Q_v$, and an output port for each edge $(v,u) \in (I \cup M) \times O$.  

%For notational convenience in the proof below, we extend the domain of $T_v$ and $T_{(v,u)}$ to $A^* \times Q_v$ by setting $T_v(a,q) = q$ and $T_{(v,u)}(a,q)=\epsilon$ whenever $a \notin A_v$, where $\epsilon$ is the empty word.

\begin{lemma}
\label{l.localglobalabelian}
\moniker{Local Abelianness Implies Global Abelianness}
If $\Net$ halts on all inputs and $\Proc_v$ is an abelian processor for each $v \in I \cup M$, then $\Proc_{I,M}$ is an abelian processor.
\end{lemma}

\begin{proof}
Given an input $\iota \in A_I^*$ and an initial state $\qq \in Q_{I \cup M}$, we can process one letter at a time to obtain a complete legal execution for $|\iota|.\qq$.  Now suppose we are given inputs $\iota, \iota'$  
such that $|\iota| = |\iota'|$.   By Lemma~\ref{l.dichotomy}, any two complete legal executions $w,w'$ for $|\iota|.\qq = |\iota'|.\qq$ satisfy $|w|=|w'|$.  In particular, $t_w(\qq) = t_{w'}(\qq)$, so the final state of $\Proc_{I,M}$ does not depend on the order of input. 

Now given $v \in I \cup M$, let $w_v$ and $w'_v$ respectively be the words obtained from $w$ and $w'$ by deleting all letters not in $A_v$.  Then $|w_v| = |w'_v|$.
\old{
Since $\Proc_v$ is an abelian processor, we have $T_v(w_v,q_v) = T_v(w'_v,q_v)$, so the final state of $\Proc_I$ does not depend on the order of input.  
}
For each edge $(v,u) \in (I \cup M) \times O$, since $\Proc_v$ is an abelian processor, 
	\[ |T_{(v,u)}(w_v,q_v)| = |T_{(v,u)}(w'_v,q_v)| \]
so for each $a \in A_u$ the number of letters $a$ sent along $(v,u)$ does not depend on the order of input.
\end{proof}

For another example of a local-to-global principle, see \cite[Lemma~2.6]{part3} (note that $I=V$ in that example: input is permitted anywhere in the network). 
%Lemma~\ref{l.localglobalirreducible}
Further local-to-global principles in the case of rotor networks are explored in \cite{GLPZ12}.

\begin{remark}
In the preceding lemma, the input to an abelian network takes the form of letters sent by the user to nodes in $I$, and the output takes the form of letters received by the user from nodes in $O$. In particular, the user has no access to the internal states of any processors, nor to letters sent or received by the mediating nodes.  In this setup, we can say that the network ``computes'' a function $\N^I \to \N^O$.  
%This notion of input and output is well suited to combining abelian networks into larger networks: for instance, given two networks $(I,M,O)$ and $(I',M',O')$ with $O=I'$, they can be combined in series to form a network with input nodes $I$, output nodes $O'$ and mediating nodes $M \cup O \cup M'$.  
%This modular viewpoint 
This notion of computation is explored in \cite{HLW16}, which identifies a set of five \emph{abelian logic gates} such that any function computable by a finite abelian processor can be computed by a finite network of abelian logic gates.

In cases when the user has access to the internal states of the processors in $I \cup O$, we can regard the input as a pair $\xx.\qq$ with $\xx \in \N^I$ and $\qq \in Q_I$, and the output as a pair $\xx'.\qq'$ with $\xx' \in \N^O$ and $\qq' \in Q_O$.  For example, in the case of a sandpile or rotor network on $\Z^2$, one might want to think of the network's output as the intricate patterns displayed by the final states of the processors. The next section describes an example when the user benefits from the ability to set up the initial states of the processors as part of the input.
\end{remark}

\old{
\begin{definition} \label{d.finite}
An abelian network $\Net$ is \emph{locally finite} if each processor $\Proc_v$ is a finite-state machine, i.e., each state space $Q_v$ and alphabet $A_v$ is finite. 

$\Net$ is \emph{spatially finite} if its underlying graph $G$ is finite.

$\Net$ is \emph{finite} if it is both locally and spatially finite.
\end{definition}
}

\subsection{Monotone integer programming}
\label{s.MIP}

In this section we describe a class of optimization problems that abelian networks can solve.
Let $A$ be a finite set and $ F : \N^A \to \N^A$ a nondecreasing function: $F(\uu) \leq F(\vv)$ whenever $\uu \leq \vv$ (inequalities are coordinatewise).  Let $\cc \in \R^A$ be a vector with all coordinates positive, and consider the following problem.
	\begin{align} \label{e.MIP} &\textrm{Minimize} && \cc^T \uu \nonumber \\ 
	&\textrm{subject to } && \uu \in \N^A \quad \mbox{ and } \quad F(\uu) \leq \uu. \end{align}
Let us call a vector $\uu \in \N^A$ \emph{feasible} if $F(\uu) \leq \uu$.  If $\uu_1$ and $\uu_2$ are both feasible, then their coordinatewise minimum is feasible:
	\[ F(\min(\uu_1,\uu_2)) \leq \min (F(\uu_1), F(\uu_2)) \leq \min(\uu_1,\uu_2). \]
Therefore if a feasible vector exists then the minimizer is unique and independent of the positive vector $\cc$: it is simply the coordinatewise minimum of all feasible vectors.

Let $\Net$ be an abelian network with finite alphabet $A$ and finite or infinite state space $Q$. Fix $\xx \in \N^A$ and $\qq \in Q$, and let
$F : \N^A \to \N^A$ be given by
	\[ F(\uu) = \xx + \mathbf{N}(\uu,\qq) \]
where $\mathbf{N}(\uu,\qq)$ is defined as $\mathbf{N}(w,\qq)$ for any word $w$ such that $|w|=\uu$. The function $F$ is well-defined and nondecreasing by Lemma~\ref{l.monotonicity}.

Recall the odometer $[\xx.\qq]$ is the vector of detailed local run times (Definition~\ref{d.odometer}).

\begin{theorem} 
\moniker{Abelian Networks Solve Monotone Integer Programs}
\label{t.MIP}
\begin{enumerate}[\em (i)]
\item If $\Net$ halts on input $\xx.\qq$, then $\uu = [\xx.\qq]$ is the unique minimizer of \eqref{e.MIP}.
\item If $\Net$ does not halt on input $\xx.\qq$, then \eqref{e.MIP} has no feasible vector $\uu$.
\end{enumerate}
\end{theorem}

\begin{proof}
By \eqref{e.piformula}, any complete execution $w$ for $\xx.\qq$ satisfies $F(|w|) \leq |w|$; and conversely, if $F(\uu) \leq \uu$ then any $w \in A^*$ such that $|w|=\uu$ is a complete execution for $\xx.\qq$.

If $\Net$ halts on input $\xx.\qq$ then the odometer $[\xx.\qq]$ is defined as $|w|$ for a complete \emph{legal} execution $w$. By the least action principle (Lemma~\ref{l.leastaction}), for any complete execution $w'$ we have $|w|_a \leq |w'|_a$ for all $a \in A$. Thus
	\[ [\xx.\qq]_a = \min \{ |w'|_a \, : \, w' \mbox{ is a complete execution for } \xx.\qq \} \]
so $[\xx.\qq]$ is the coordinatewise minimum of all feasible vectors.

If $\Net$ does not halt on input $\xx.\qq$, then there does not exist a complete execution for $\xx.\qq$, so there is no feasible vector.
\end{proof}

For any nondecreasing $F : \N^A \to \N^A$, there is an abelian network $\Net_F$ that solves the corresponding optimization problem \eqref{e.MIP}.
Its underlying graph is a single vertex $v$ with a loop $e=(v,v)$. It has state space $Q = \N^A$, transition function $T_v(a,\qq) = \qq+1_a$ and message passing function satisfying
	\[ |T_e(a,\qq)| = F(\qq+1_a) - F(\qq) \] 
for all $a \in A$ and $\qq \in Q$. For the input we take $\xx=F(\zero)$ and $\qq=\zero$.

\begin{remark} %%% to refer to from part 3
In general the problem \eqref{e.MIP} is nonlinear, but in the special case of a \emph{toppling network} 
%(enlarged to allow negative chip counts as described in \textsection\ref{s.toppling}) 
it is equivalent to a linear integer program of the following form.
	\begin{align} \label{e.IP} &\textrm{Minimize} && \cc^T \vv \nonumber \\ 
	&\textrm{subject to } && \vv \in \N^A \quad \mbox{ and } \quad L \vv \geq \bb. \end{align}
Here $\cc \in \R^A$ has all coordinates positive; $L$ is the Laplacian matrix (\textsection\ref{s.toppling}); and $\bb = \xx - \rr + \one$ where $\xx_a$ is the number of chips input at $a$ and $\rr_a$ is the threshold of $a$.  The coordinate $\vv_a$ of the minimizer is the number of times $a$ topples.  To see the equivalence of \eqref{e.MIP} and \eqref{e.IP} for toppling networks, note that $F$ takes the following form for a toppling network:
	\[ F(\uu) = \xx + (D-L) \floor{D^{-1} \uu} \]
where $D$ is the diagonal matrix with diagonal entries $\rr_a$, and $\floor{\cdot}$ denotes the coordinatewise greatest integer function.
Using that $D-L$ is a nonnegative matrix, one checks that $\uu = \xx + (D-L)\vv$ is feasible for \eqref{e.MIP} if and only if $\vv$ is feasible for \eqref{e.IP}. 
\old{%%%
Indeed, $F(\uu) \leq \uu$ if and only if
	\[ \xx + (D-L) \floor{D^{-1}(\xx - (D-L)\vv))} \leq \xx + (D-L)\vv \]
which happens if and only if
	\[ \floor{\vv + D^{-1}(\xx-L\vv)} \leq \vv \]
which happens if and only if
	\[ x - L\vv \leq \rr-\one. \]
}%%%
We remark that for a \emph{general} integer matrix the problem of whether \eqref{e.IP} has a feasible vector $\mathbf{v}$ is $\NP$-complete (see, for example, \cite{Pap81}) but that the Laplacian $L$ for an abelian network is constrained to have off-diagonal entries $\leq 0$.  See \cite{FL16} for a discussion of the computational complexity of this problem when $L$ is a directed graph Laplacian.
\end{remark}

\section{Concluding Remarks}

We indicate here a few directions for further research on abelian networks. Other directions are indicated in the sequels \cite{part2,part3}.

%\silentsubsec{Asynchronous graph algorithms}
\subsection{Asynchronous graph algorithms}
\label{s.graphalg}

%What does an abelian network ``know'' about its underlying graph?
%For instance, 
Chan, Church and Grochow \cite{CCG14} have shown that a rotor network can detect whether its underlying graph is planar (with edge orderings respecting the planar embedding).
Theorem~\ref{t.wishlist} shows that abelian networks can compute asynchronously, and Theorem~\ref{t.MIP} gives an example of something they can compute.
It would be interesting to explore whether abelian networks can perform computational tasks like shortest path, pagerank, image restoration and belief propagation.  We note one practical deficiency of abelian networks: In the words of an anonymous referee, ``determining when a network has finished computing requires some computational overhead'' outside the network.

%\silentsubsec{Abelian networks with shared memory}
\subsection{Abelian networks with shared memory}

In \textsection \ref{s.cellular} we have emphasized that abelian networks do not rely on shared memory.  Yet there are quite a few examples of processes with a global abelian property that do.  Perhaps the simplest is \emph{sorting by adjacent transpositions}: suppose $G$ is a finite path and each vertex $v$ has state space $Q_v = \Z$.  The processors now live on the edges: for each edge $e=(v,v+1)$ the processor $\Proc_e$ acts by swapping the states $q(v)$ and $q(v+1)$ if $q(v)>q(v+1)$.  This example does not fit our definition of abelian network because the processors of edges $(v-1,v)$ and $(v,v+1)$ share access to the state $q(v)$.   Indeed, from our list of five goals in \textsection \hyperlink{i.wishlist}{2} this example satisfies items (a)--(c) only: The final output is always sorted, and the run time does not depend on the execution, but the local run times do depend on the execution.  For instance, when $G$ is a path with three vertices and two edges $s_1$ and $s_2$, both $s_1 s_2 s_1$ and $s_2 s_1 s_2$ are complete legal executions for the initial state $(3,2,1)$. The edge $s_1$ performs two swaps in the first execution, but only one swap in the second execution.

What is the right definition of an abelian network with shared memory? Examples could include the numbers game of Mozes \cite{Moz90}, $k$-cores of graphs and hypergraphs, Wilson cycle popping \cite{Wil96} and its extension by Gorodezky and Pak \cite{GP14}, source reversal \cite{GP00} and cluster firing \cite{notes,Bac12,CPS12}.

%\silentsubsec{Nonabelian networks}
\subsection{Nonabelian networks}

The work of Krohn and Rhodes \cite{KR65, KR68} led to a detailed study of how the algebraic structure of monoids relates to the computational strength of corresponding classes of automata. It would be highly desirable to develop such a dictionary for classes of automata \emph{networks}. Thus one would like to weaken the abelian property and study networks of solvable automata, nilpotent automata, etc.  Such networks are nondeterministic --- the output depends on the order of execution --- so their theory promises to be rather different from that of abelian networks.  
It could be fruitful to look for networks that exhibit only limited nondeterminism. A concrete example 
is a sandpile network with annihilating particles and antiparticles, studied by Robert Cori (unpublished) and in \cite{CPS12} under the term ``inverse toppling.''

\silentsec{Acknowledgments}

The authors thank Spencer Backman, Olivier Bernardi, Deepak Dhar, Anne Fey, Sergey Fomin, Christopher Hillar, Michael Hochman, Alexander Holroyd, Benjamin Iriarte, Mia Minnes, Ryan O'Donnell, David Perkinson, James Propp, Leonardo Rolla, Farbod Shokrieh, Allan Sly and Peter Winkler for helpful discussions. We thank two anonymous referees for insightful comments that improved the paper.

This research was supported by an NSF postdoctoral fellowship and NSF grants DMS-1105960 and DMS-1243606, and by the UROP and SPUR programs at MIT.

\end{document}